\pdfoutput=1
\RequirePackage{fix-cm}
\documentclass[a4paper,UKenglish]{lipics}
\hypersetup{%
   breaklinks,%
   colorlinks=true,%
   linkcolor=[rgb]{0.45,0.0,0.0},%
   urlcolor=[rgb]{0.05,0.390,0.0},%
   citecolor=[rgb]{0,0,0.45}
}%

\usepackage{microtype}
\usepackage{lmodern} 

\usepackage{amssymb,latexsym,amsthm}
\usepackage{adjustbox}
\usepackage{algorithm}
\usepackage{enumerate}
\usepackage{appendix}
\usepackage{amsmath}
\usepackage[numbers]{natbib}

\usepackage{mathtools}

\usepackage{etoolbox}
\makeatletter
    \pretocmd{\NAT@citexnum}{\@ifnum{\NAT@ctype>\z@}{\let\NAT@hyper@\relax}{}}{}{}
\makeatother

\usepackage{xspace} 
\usepackage{color}%

\usepackage{subcaption}
\usepackage{pgf,tikz}
\usepackage{wrapfig}  
\usepackage{cutwin}  

\usepackage[noend]{algorithmic}
\usetikzlibrary{arrows,shapes,calc,intersections,through,backgrounds,patterns}
\usepackage{thmtools, thm-restate}

\theoremstyle{plain}

\newtheorem{claim}{Claim}[section]
\newtheorem{defn}{Definition}[section]
\newtheorem{obs}[defn]{Observation}

\newcommand{\Real}{\ensuremath{\mathbb{R}}\xspace}%

\newcommand{\eps}{\ensuremath{\varepsilon}}
\newcommand{\cost}{\ensuremath{\text{cost}}}
\newcommand{\diam}{\ensuremath{\text{diam}}}
\newcommand{\dmin}{\ensuremath{d_{\text{min}}}}
\newcommand{\dmax}{\ensuremath{d_{\text{max}}}}
\newcommand{\opt}{\ensuremath{{\text{opt}}}}
\newcommand{\Q}{\ensuremath{{\cal Q}}}
\newcommand{\level}{\ensuremath{\text{level}}}
\newcommand{\B}{\ensuremath{{\cal B}}}
\newcommand{\Sp}{\ensuremath{{\cal S'}}}

\DeclareMathOperator*{\argmin}{arg\,min}

\begin{document}
\title{Approximate Clustering via Metric Partitioning\footnote{This material is based upon work supported by
the National Science Foundation under Grant CCF-1318996}}

\author{Sayan Bandyapadhyay\thanks{sayan-bandyapadhyay@uiowa.edu}} 
\author{Kasturi Varadarajan\thanks{kasturi-varadarajan@uiowa.edu}}
\affil{
  Department of Computer Science\\
  University of Iowa, Iowa City, USA\\
  }
  
\authorrunning{S.\,Bandyapadhyay and K.\,Varadarajan} 

\Copyright{Sayan Bandyapadhyay and Kasturi Varadarajan}

\maketitle

\begin{abstract}
In this paper we consider two metric covering/clustering problems - \textit{Minimum Cost Covering Problem} (MCC) and $k$-clustering. In the MCC problem, we are given two point sets $X$ (clients) and $Y$ (servers), and a metric on $X \cup Y$. We would
like to cover the clients by balls centered at the servers. The objective 
function to minimize is the sum of the $\alpha$-th power of the radii of
the balls. Here $\alpha \geq 1$ is a parameter of the problem (but not of
a problem instance). MCC is closely related to the $k$-clustering problem. The main difference between $k$-clustering and MCC is that in $k$-clustering one needs to select $k$ balls to cover the clients.

For any $\eps > 0$, we describe quasi-polynomial time $(1 + \eps)$ approximation algorithms for both of the problems. However, in case of $k$-clustering the algorithm uses $(1 + \eps)k$ balls. Prior to our work, a $3^{\alpha}$ and a ${c}^{\alpha}$ 
approximation were achieved by polynomial-time algorithms for MCC and $k$-clustering, respectively, where $c > 1$ is an absolute constant. These two problems are thus interesting examples of metric covering/clustering problems that admit $(1 + \eps)$-approximation (using $(1+\eps)k$ balls in case of $k$-clustering), if
one is willing to settle for quasi-polynomial time. In contrast, for the variant of MCC where $\alpha$ is part of the input, we show under standard assumptions that no polynomial time algorithm can achieve an approximation factor better than $O(\log |X|)$ for $\alpha \geq \log |X|$. 
%
\end{abstract} 

\section{Introduction}\label{sec:intro}

We consider two metric covering/clustering problems. In the first problem, we are given two point sets $X$ (clients) and $Y$ (servers), and a metric $d$ on $X \cup Y$. For $z \in X \cup Y$
and $r \geq 0$, the ball $B(z,r)$ centered at $z$ and having radius $r \geq 0$ is
the set $\{y \in X \cup Y | d(z,y) \leq r\}$. A \textit{cover} for a subset $P \subseteq X$ is a set of balls, each centered at a point of $Y$, whose union contains $P$. The \textit{cost} of a set $\mathcal{B}=\{B_1,\ldots,B_k\}$ of balls, denoted by $\cost(\mathcal{B})$,  is $\sum_{i=1}^k {r(B_i)}^\alpha$, where $r(B_i)$ is the radius of $B_i$, and $\alpha \geq 1$ is a parameter of the problem (but not of
a problem instance). The goal is to compute a minimum cost \textit{cover} for 
the clients $X$. We refer to this problem as the \textit{Minimum Cost Covering Problem} (MCC).

In the second problem, we are given a set $X$ of $n$ points, a metric $d$ on $X$, and a positive integer $k$. Unlike in the case of MCC, here each ball is centered at a point in $X$.\footnote{Our results do generalize to the problem where we distinguish between clients and servers as in the MCC.} The \textit{cost} $\cost(\mathcal{B})$ of a set $\mathcal{B}$ of balls is defined exactly in the same way as in the case of MCC. The goal is to find a set $\mathcal{B}$ of $k$ balls whose union contains all the points in $X$ and $\cost(\mathcal{B})$ is minimized. We refer to this problem as $k$-clustering.

Inspired by applications in wireless networks, MCC has been well studied \cite{Lev-TovP05}. One can consider the points in $Y$ as the potential locations of mobile towers and the points in $X$ as the locations of customers. A tower can be configured in a way so that it can serve the customers lying within a certain distance. But the service cost increases with the distance served. The goal is to serve all the customers minimizing the total cost. For modelling the energy needed for wireless transmission, it is common to consider the value of $\alpha$ to be at least 1. 

For the MCC problem with $\alpha = 1$, a primal-dual algorithm of Charikar and Panigrahy \cite{CharikarP04} leads to an approximation guarantee of $3$; their result generalizes to $\alpha \geq 1$, with an approximation guarantee of $3^{\alpha}$. The problem
is known to be NP-hard for $\alpha > 1$, even when $X$ and $Y$ are points
in the Euclidean plane \cite{AltABEFKLMW06}. 
The case $\alpha=1$ has received particular attention. The first PTAS for the Euclidean plane was designed by Lev-Tov and Peleg \cite{Lev-TovP05}. Later, Gibson~{\em et.~al} \cite{GibsonKKPV12} have designed a polynomial time exact algorithm for this problem when $X$ and $Y$ are points in the plane, and the underlying distance function $d$ is
either the $l_1$ or $l_{\infty}$ metric. For the $l_2$ metric they also get an exact algorithm if one assumes two candidate solutions can be compared efficiently; without this assumption, they get a $(1 + \eps)$ approximation. Their algorithm is based on a separator theorem that, for any optimal solution, proves the existence of a balanced separator that intersects with at most 12 balls in the solution. In a different work they have also extended the exact algorithm to arbitrary metric spaces \cite{GibsonKKPV10}. The running time is quasi-polynomial if the aspect ratio of the metric (ratio of maximum to minimum interpoint distance) is bounded by a polynomial in the number of points. When the aspect ratio is not polynomially bounded, they obtain a $(1 + \eps)$ approximation in quasi-polynomial time. Their algorithms are based on
a partitioning of the metric space that intersects a small number of balls in the
optimal cover. 


When $\alpha > 1$, the structure that holds for $\alpha = 1$ breaks down. It is no longer the case, even in the Euclidean plane, that there is a good separator
(or partition) that intersects a small number of balls in an optimal solution.
In the case $\alpha=2$ and the Euclidean plane, the objective function models the total area of the served region, which arises in many practical applications. Hence this particular version has been studied in a series of works. Chuzhoy developed an unpublished 9-factor approximation algorithm for this version. Freund and Rawitz \cite{FreundR03} present this algorithm and give a primal fitting interpretation of the approximation factor. Bilo~{\em et. al}\ \cite{BiloCKK05} have extended the techniques of Lev-Tov and Peleg \cite{Lev-TovP05} to get a PTAS that works for any $\alpha \geq 1$ and for any fixed dimensional Euclidean space. The PTAS is based on a sophisticated use of the \textit{shifting strategy} which is a popular technique in computational geometry for solving problems in $\mathbb{R}^d$ \cite{ErlebachJS05,HochbaumM85}. For general metrics, however,  the best known approximation guarantee for $\alpha > 1$ remains the already
mentioned $3^{\alpha}$ \cite{CharikarP04}. 

The $k$-clustering problem has applications in many fields including Data Mining, Machine Learning and Image Processing. 
Over the years it has been studied extensively from both theoretical and practical perspectives \cite{BiloCKK05,CharikarP04,DoddiMRTW00,GibsonKKPV10,GibsonKKPV12,Schaeffer07}. The problem can be seen as a variant of MCC where $Y=X$ and at most $k$ balls can be chosen to cover the points in $X$. As one might think, the constraint on the number of balls that can be used in $k$-clustering makes it relatively harder than MCC. Thus all the hardness results for MCC also hold for $k$-clustering. For $\alpha = 1$, Charikar and Panigrahy \cite{CharikarP04} present a polynomial time algorithm with an approximation guarantee of about $3.504$. Gibson~{\em et.~al} \cite{GibsonKKPV10,GibsonKKPV12} obtain the same results for $k$-clustering with $\alpha=1$ as the ones
described for MCC, both in $\mathbb{R}^d$ and arbitrary metrics. Recently, Salavatipour and Behsaz \cite{BehsazS12} have obtained a polynomial time exact algorithm for $\alpha=1$ and metrics of unweighted graphs, if we assume that no singleton clusters are allowed. 
However, in case of $\alpha > 1$ the best known approximation factor (in polynomial time) for general metrics is $c^{\alpha}$, for some absolute constant $c > 1$; this follows from the analysis of Charikar and Panigrahy \cite{CharikarP04}, who explicitly study only the case $\alpha = 1$. In fact, no better polynomial time approximation is known even for the Euclidean plane. We note that though the polynomial time algorithm in \cite{BiloCKK05} yields a $(1+\eps)$ approximation for $k$-clustering in any fixed dimensional Euclidean space and for $\alpha \geq 1$, it can use $(1+\eps)k$ balls.

In addition to $k$-clustering many other clustering problems ($k$-means, $k$-center, $k$-median etc.) have been well studied 
\cite{BandyapadhyayV16,VPC16,FRS16,har2011geometric}. 

In this paper we address the following interesting question. Can the techniques employed by \cite{BiloCKK05} for fixed dimensional Euclidean spaces be generalized to give $(1 + \eps)$ approximation for MCC and $k$-clustering in any metric space? Our motivation for studying the problems in a metric context is partly that it includes two geometric contexts: (a) high dimensional Euclidean spaces; and (b) shortest path distance metric in the presence of polyhedral obstacles in $\Real^2$ or $\Real^3$. 


\subsection{Our Results and Techniques}
In this paper we consider the metric MCC and $k$-clustering with $\alpha \geq 1$. For any $\epsilon > 0$, we design a $(1+\epsilon)$-factor approximation algorithm for 
MCC that runs in quasi-polynomial time, that is, in $2^{(\log mn/\eps)^c}$ time, where $c > 0$ is a constant, $m = |Y|$, and $n = |X|$. We also have designed a similar algorithm for $k$-clustering that uses at most $(1+\eps)k$ balls and yields a solution whose cost is at most $(1+\eps)$ times the cost of an optimal $k$-clustering solution. The time complexity of the latter algorithm is also quasi-polynomial. As already noted, somewhat stronger guarantees are already known for the case $\alpha = 1$ of these problems \cite{GibsonKKPV10}, but the structural properties that hold for $\alpha = 1$ make it rather special.

The results in this paper should be
compared with the polynomial time algorithms \cite{CharikarP04} that guarantee
$3^{\alpha}$ approximation for MCC and ${c}^{\alpha}$ approximation for $k$-clustering. The MCC and $k$-clustering are thus interesting examples of metric covering/clustering problems that admit $(1 + \eps)$-approximation (using $(1+\eps)k$ balls in case of $k$-clustering), if
one is willing to settle for quasi-polynomial time. From this perspective, our results are surprising, as most of the problems in general metrics are APX-hard. The MCC and $k$-clustering are also examples where the techniques used in fixed dimensional Euclidean spaces generalize nicely to metric spaces. This is in contrast to the {\em facility location problem} \cite{Arora}.

The algorithms that we have designed for both of the problems use similar techniques that exploit the following key property of optimal covers: there are only
a ``small'' number of balls whose radius is ``large''. We can therefore afford to guess
these balls by an explicit enumeration. However, there can be a ``large'' number of balls with ``small'' radius. To help `find' these, we partition the metric space
into blocks (or subsets) with at most half the original diameter, and recurse on each block.
We have to pay a price for this recursion in the approximation guarantee. This
price depends on the number of blocks in the partition that a small radius
ball can intersect. (This is not an issue in the case $\alpha = 1$, where each ball that is not guessed intersects precisely one of the blocks \cite{GibsonKKPV10}.)

We are led to the following problem: is there a way to probabilistically 
partition a metric space into blocks of at most half the diameter, so that
for any ball with ``small'' radius, the expected number of blocks that intersect
the ball can be nicely bounded? 
The celebrated partitioning algorithms of 
Bartal \cite{Bartal} and  Fakcharoenphol, Rao, and Talwar \cite{FRT} guarantee that the probability 
that such a ball is intersected by two or more blocks is nicely bounded.
However, their bounds on the probability that a small ball is intersected do not directly imply a good bound on the expected number of blocks intersected by a small ball. Indeed, if one employs the partitioning algorithm of \cite{FRT}, the expected number of blocks intersected by a small ball can be quite ``large''
. Fortunately, the
desired bound on the expectation can be shown to hold for the algorithm of Bartal \cite{Bartal}, even though he did not study the expectation itself. We use a similar partitioning scheme and derive the expectation bound in Section~\ref{sec:partition}, using an analysis that closely tracks previous work \cite{ abraham2006advances,bartal2004graph, kamma2014cutting}. While the bound on the expectation is easily derived from previous work, our work is the first to study and fruitfully apply this bound.  

The algorithms for MCC and $k$-clustering, which use the partitioning scheme of Section~\ref{sec:partition}, are described in Section~\ref{sec:MCC} and \ref{sec:clustering}, respectively.
In Section \ref{sec:inapprox}, we consider the approximability of a variant of the MCC where we allow $\alpha$ to be part of the input.
For $\alpha \geq \log |X|$, we show, under standard complexity theoretic assumptions, that no polynomial (or quasi-polynomial) time algorithm for MCC can achieve an approximation factor better than $O(\log |X|)$. This partly explains the dependence on $\alpha$ of the running time of our algorithms.

\section{The Partitioning Scheme}\label{sec:partition}
Let $Z$ be a point set with an associated metric $d$, let $P \subseteq Z$
be a point set with at least $2$ points, and $n \geq |P|$ be a parameter. For $Q \subseteq Z$, denote the maximum interpoint distance (or diameter) of $Q$ by $\diam(Q)$. Consider any partition of $P$ into subsets (or blocks) $\{P_1, P_2, \ldots, P_t\}$, where $2\leq t\leq |P|$. Abusing notation, we will also view $\{P_1, P_2, \ldots, P_t\}$ as a sequence of blocks. We say that $P_i$ non-terminally (resp. terminally) intersects a ball $B$ if $P_i$ intersects $B$ and it is not (resp. it is) the last set in the sequence $P_1, P_2, \ldots, P_t$ that intersects $B$. We would like to find a partition $\{P_1, P_2, \ldots, P_t\}$ of $P$ that ensures the following properties:


\begin{enumerate}
 \item For each $1 \leq i \leq t$, $\diam(P_i) \leq$ $\diam(P)/2$.
 
 \item For any ball $B$ (centered at some point in $Z$) of radius $r \leq \frac{\diam(P)}{16 \log n}$,
the expected size of the set $\{ i | P_i \cap B \neq \emptyset \}$ is
at most $1 + c \frac{r}{\diam(P)} \log n$, where $c > 0$ is a constant. In other words, the expected number
of blocks in the partition that intersect $B$ is at most $1 + c \frac{r}{\diam(P)} \log n$.

\item For any ball $B$ (centered at some point in $Z$) of radius $r \leq \frac{\diam(P)}{16 \log n}$, the expected number
of blocks in the partition that non-terminally intersect $B$ is at most $c \frac{r}{\diam(P)} \log n$, where $c > 0$ is a constant.
\end{enumerate}

We note that the second property follows from the third, as the number of blocks that intersect ball $B$ is at most one more than the number of blocks that non-terminally intersect $B$. 
We will design a probabilistic partitioning algorithm that finds a partition with the desired properties. 

\subsubsection{The Partitioning Scheme of \cite{FRT}}
We first explain why the probabilistic partitioning algorithm of \cite{FRT} does
not achieve this guarantee. In this algorithm, we first pick a $\beta$ uniformly
at random from the interval $[\frac{\delta}{8}, \frac{\delta}{4}]$, where
$\delta = \diam(P)$. Also, let $\pi_1, \pi_2, \ldots, \pi_p$ be a permutation of
$P$ chosen uniformly at random. We compute $P_1, P_2, \ldots, P_p$ in order as follows. Suppose we already computed $P_1, \ldots, P_{i-1}$. We let
\[ P_i = \{x \in P \setminus (P_1 \cup P_2 \cup \cdots \cup P_{i-1}): d(x,\pi_i) \leq \beta \}.\]
We will refer to $P_i$ as $\pi_i$'s cluster. We return the partition $\{P_i \ | \ P_i \neq \emptyset\}$.

Consider the following weighted tree. Let vertex $u$ be connected to vertices
$u_1, u_2, \ldots, u_b$ using edges of weight $\frac{\delta}{16 \log n}$. Here,
$n$, $\delta$, and $b$ are parameters. Let $V_1, V_2, \ldots, V_b$ be disjoint
sets with $b$ vertices each. For each $i$, $u_i$ is connected to every vertex
in $V_i$ using edges of weight $\frac{\delta}{4} - \frac{\delta}{16 \log n}$. Finally, $z$ is a new vertex that is connected to $u$ using an edge of weight $\frac{3\delta}{4}$. Consider the metric induced by this weighted graph, and let
$P$ denote the vertex set. That is $P = \{u\} \cup \bigcup_i \{u_i\} \cup
\bigcup_i V_i \cup \{z\}$. Set $n:= |P|$ and note that $b = \Theta(\sqrt{n})$. Also, $\delta = \diam(P)$.  

Let $B = B(u,r)$ where $r = \frac{\delta}{16 \log n}$. Notice that the ball $B$
consists of the points $\{u, u_1, u_2, \ldots, u_b \}$. Consider running the
probabilistic partitioning algorithm of \cite{FRT}, described above, on $P$. We
argue that the expected number of blocks in the output partition that intersect
$B$ is $\Omega(\frac{\sqrt{n}}{\log n})$, which is asymptotically larger than
$1 + c \frac{r}{\diam(P)} \log n = O(1)$. 

Fix $1 \leq i \leq b$. We observe that in the output partition, $u_i$ belongs 
to the cluster of some point in $\{u, u_1, u_2, \ldots, u_b\}$ or of some point
in $V_i$. Call $u_i$ {\em isolated} if $u_i$ belongs to the cluster of some point in $V_i$. If $u_i$ is isolated, the cluster containing $u_i$ does not contain
any of the other $u_j$. Thus, the number of blocks of the output partition that
intersect $B$ is at least the number of vertices in $\{u_1, u_2, \ldots, u_b \}$that are isolated.

Note that $u_i$ is isolated if the following two events occur: (a) $\beta
\in [\frac{\delta}{4} - \frac{\delta}{16 \log n}, \frac{\delta}{4}]$; (b) some vertex in
$V_i$ appears in the random permutation $\pi_1, \pi_2, \ldots, \pi_p$ before
all vertices in $\{u, u_1, \ldots, u_b\}$. The probability of these two events
occuring is $\Omega(\frac{1}{\log n})$. It follows that the expected number
of isolated vertices, and thus the expected number of blocks in the partition
that intersect $B$, is $\Omega(\frac{b}{\log n}) = \Omega(\frac{\sqrt{n}}{\log n})$.   

\subsubsection{Probability Distribution}
Before describing our partitioning algorithm, we consider a probability distribution that it uses. Given a positive real $\delta$ and an integer $k \geq 2$, the distribution is denoted by dist$(\delta,k)$. The probability density function (pdf) $f$ of dist$(\delta,k)$ is the following:

$$f(x) =
\left\{
	\begin{array}{ll}
		0  & \mbox{if } x < \delta/8 \text{ and } x > \delta/4\\
		\frac{8\log k}{\delta} \frac{1}{2^{i}} & \mbox{if } \frac{\delta}{8}+(i-1)\frac{\delta}{8\log k} \leq x < \frac{\delta}{8}+i\frac{\delta}{8\log k} \text{ for } 1\leq i\leq \log k-1\\
		\frac{8\log k}{\delta} \frac{2}{k} & \mbox{if } \frac{\delta}{4}-\frac{\delta}{8\log k} \leq x \leq \frac{\delta}{4}
	\end{array}
\right.$$

The following observation shows that $f$ is indeed a density function.

\begin{obs}
$f$ is a probability density function. 
\label{obs:density}
\end{obs}

\begin{proof}
It is sufficient to show that $f$ satisfies the two properties of density function. As $\delta$ and $k$ are nonnegative it is easy to see that $f(x) \geq 0$ for $x \in (-\infty,+\infty)$. Also,
$$
\begin{aligned}
\int_{-\infty}^{\infty} f(x)dx & =\int_{\frac{\delta}{8}}^{\frac{\delta}{8}+\frac{\delta}{8\log k}} \frac{8\log k}{\delta} \frac{1}{2}dx +\int_{\frac{\delta}{8}+\frac{\delta}{8\log k}}^{\frac{\delta}{8}+\frac{2\delta}{8\log k}} \frac{8\log k}{\delta} \frac{1}{2^{2}}dx\\ & +\ldots +\int_{\frac{\delta}{8}+\frac{(\log k-2)\delta}{8\log k}}^{\frac{\delta}{8}+\frac{(\log k-1)\delta}{8\log k}} \frac{8\log k}{\delta} \frac{1}{2^{\log k -1}}dx+\int_{\frac{\delta}{4}-\frac{\delta}{8\log k}}^{\frac{\delta}{4}} \frac{8\log k}{\delta} \frac{2}{k}dx\\
& =\frac{1}{2} \sum_{i=0}^{\log k -2} \frac{1}{2^{i}}+ \frac{2}{k} = 1-{(\frac{1}{2})}^{\log k -1}+ \frac{2}{k}=1
\end{aligned}$$
\end{proof}

Consider the interval $[\frac{\delta}{8},\frac{\delta}{4}]$. Now divide the interval into $\log k$ subintervals of equal length. The $i^{th}$ interval for $1\leq i\leq \log k -1$ is defined as $[\frac{\delta}{8}+(i-1)\frac{\delta}{8\log k},\frac{\delta}{8}+i\frac{\delta}{8\log k})$. The last interval is $[\frac{\delta}{4}-\frac{\delta}{8\log k}, \frac{\delta}{4}]$. Denote the $j^{th}$ interval by $I_j$ for $1\leq j\leq \log k$. 

To sample a $\beta$ according dist$(\delta,k)$, we first pick one of these intervals from the distribution that assigns a probability of $1/2^j$ to $I_j$ for
$1 \leq j \leq \log k - 1$, and a probability of $2/k$ to $I_{\log k}$. Having
picked an interval $I_j$, we generate $\beta$ uniformly at random from it.

Now we discuss a vital property of the distribution dist$(\delta,k)$ which we use in the analysis of the partitioning algorithm. For an event $E$, let Pr$[E]$ denotes the probability that $E$ occurs. 
Now consider the following random process. We sample a value $p$ from dist$(\delta,k)$. Let $E_j$ denotes the event that $p\in I_j$. Then we have the following observation.

\begin{obs}\label{obs:sumeq}
Pr$[E_j] = \sum_{i=j+1}^{\log k} Pr[E_i] \ $ for $1\leq j\leq \log k -1$.
\end{obs}

\begin{proof}
The proof follows from the definition of the pdf of dist$(\delta,k)$. 
$$
\begin{aligned}
\text{Pr}[E_j]& =\int_{\frac{\delta}{8}+(j-1)\frac{\delta}{8\log k}}^{\frac{\delta}{8}+j\frac{\delta}{8\log k}} \frac{8\log k}{\delta} \frac{1}{2^{j}}dx = \frac{1}{2^{j}} = \frac{1}{2^{j}} -\frac{2}{k} + \frac{2}{k} = \frac{1}{2^{j}}(1-\frac{2^{j+1}}{k})+ \frac{2}{k}\\ & = \frac{1}{2^{j}}(1-\frac{1}{2^{\log k -j-1}})+ \frac{2}{k}=\frac{1}{2^{j+1}}\frac{(1-\frac{1}{2^{\log k -j-1}})}{\frac{1}{2}}+ \frac{2}{k}\\ &=\frac{1}{2^{j+1}}(\frac{1}{2^0}+\frac{1}{2^1}+\ldots+\frac{1}{2^{\log k -j-2}})+ \frac{2}{k}\\ &=\frac{1}{2^{j+1}}+\ldots+\frac{1}{2^{\log k -1}}+ \frac{2}{k}\\
& =\text{Pr}[E_{j+1}]+\ldots+\text{Pr}[E_{\log k -1}]+\text{Pr}[E_{\log k}]=\sum_{i=j+1}^{\log k} Pr[E_i]
\end{aligned}
$$\end{proof}

\subsubsection{Partitioning Algorithm}
Now we move on to the partitioning algorithm, which we recall, is given $Z$, 
a metric $d$ on $Z$, $P \subseteq Z$, and a parameter $n \geq |P|$. The procedure RAND-PARTITION$(P)$ described below (as Algorithm \ref{alg:partition}) takes a point set $P\subseteq Z$ as input and outputs a partition with at most $|P|$ subsets. Suppose that $P=\{p_1,\ldots,p_{|P|}\}$. The algorithm then generates 
$P_1, P_2, \ldots, P_{|P|}$ in order via the for loop. Suppose that $P_1, P_2, \ldots, P_{i-1}$ have already been constructed, and $Q = P \setminus (P_1 \cup P_2 \cup \cdots \cup P_{i-1})$. To construct $P_i$, the procedure samples a $\beta_i$ from dist$(\text{diam}(P),n)$. The choice of $\beta_i$ is done independently of the
choices of the other points. Then $P_i$ is set to $\{x \in Q \ | \ d(x,p_i) \leq \beta_i\}$. Note that this is done in the $i$'th iteration of the for loop. 
(Note that $p_i$ might not be assigned to $P_i$, as it could already be assigned to some other subset.)


\begin{algorithm}[hbt]
    \caption{RAND-PARTITION$(P)$}
    \label{alg:partition}
    \begin{algorithmic}[1]
        \REQUIRE A subset $P=\{p_1,\ldots,p_{|P|}\} \subseteq Z$        
        \ENSURE  A partition of $P$     
        \STATE $p \leftarrow |P|$
        \STATE $Q \leftarrow P$
        \FOR {$i = 1$ to $p$}
	    \STATE sample a $\beta_i$ from dist$(\text{diam}(P),n)$ corresponding to $p_i$
            \STATE $P_i \leftarrow \{x\in Q|d(x,p_i) \leq \beta_i \}$	
	    \STATE $Q \leftarrow Q\setminus P_i$
        \ENDFOR
        \RETURN $\{P_i | P_i \neq \emptyset \text{ and } 1\leq i\leq p\}$
    \end{algorithmic}
\end{algorithm}


We show that RAND-PARTITION$(P)$ satisfies the two guarantees mentioned before. To see that the first guarantee, note that the $\beta$ values are chosen from the distribution dist$(\text{diam}(P),n)$ which ensures that $\beta_i \leq \text{diam}(P)/4$ for $1\leq i\leq p$. Now each point in a subset $P_i$ is at a distance at most $\beta_i$ from $p_i$. Thus by triangle inequality $\text{diam}(P_i)\leq \text{diam}(P)/2$. In the next lemma we will show that the second guarantee also 
holds. Before that we have a definition.

We say that $P_i$ non-terminally (resp. terminally) intersects $B(y,r)$ if $P_i$ intersects $B(y,r)$ and it is not (resp. it is) the last set in the sequence  $P_1, P_2, \ldots, P_p$ that intersects $B(y,r)$. 

\begin{lemma}\label{lem:part}
(Partitioning Lemma) There is a constant $c$ such that for any ball $B(y,r)$ with $r \leq \frac{\text{diam}(P)}{{16\log n}}$, the expected number
of blocks in the output of \text{RAND-PARTITION}$(P)$ that intersect $B(y,r)$  
is at most $1 + c\frac{r}{\text{diam}(P)} \log n$. Moreover, the expected number of blocks that intersect the ball non-terminally is at most $c\frac{r}{\text{diam}(P)} \log n$.
\end{lemma}

The intuition for the lemma is as follows. Consider the beginning of the $i$'th 
iteration of the for loop and assume that ball $B(y,r)$ is not fully contained
in the union of the previously constructed blocks $P_1, \ldots, P_{i-1}$. Then,
considering the choice of $\beta_i$, the probability that $B(p_i,\beta_i)$ 
fully contains the ball $B(y,r)$ is nearly as large as the probability that
$B(p_i, \beta_i)$ intersects $B(y,r)$. If $B(p_i,\beta_i)$ 
fully contains the ball $B(y,r)$, then of course none of the blocks $P_{i+1}, P_{i+2}, \ldots, P_p$ intersect $B(y,r)$. We now proceed to the proof.
   
\begin{proof}
For a point $x \in Z$ and subset $Q \subseteq Z$, let $\dmin(x,Q) = \min_{q \in Q} d(x,q)$ and $\dmax(x,Q) = \max_{q \in Q} d(x,q)$. Fix the ball $B(y,r)$ with
$r \leq \frac{\text{diam}(P)}{{16\log n}}$.
For each $1\leq i\leq p$, consider the indicator random variable $T_i$ defined as follows:
$$T_i =
\left\{
	\begin{array}{ll}
		1  & \mbox{if } P_i \text{ intersects } B(y,r)\\
		0 & \mbox{otherwise} 
	\end{array}
\right.$$
Let the random variable $T=\sum_{i=1}^p T_i$ be the number of subsets that the ball intersects. Then $E[T]=\sum_{i=1}^p E[T_i]= \sum_{i=1}^p Pr[P_i \text{ intersects } B(y,r)]$.

Clearly,
there is at most one $P_i$ that is the last one that intersects $B(y,r)$. Thus,
\[ \sum_{i=1}^p  Pr[P_i \text{ intersects } B(y,r)] \leq 1 + 
   \sum_{i=1}^p Pr[P_i \text{ non-terminally intersects } B(y,r)]. \]

Let $x_i = \dmin(p_i,B(y,r))$ and $y_i = \dmax(p_i, B(y,r))$. By the
triangle inequality, $y_i - x_i \leq 2r$. Denote by $(S^i)$ the event 
that $\beta_i$ lands in the interval $[x_i,y_i]$.  
Note that for $P_i$ to non-terminally intersect $B(y,r)$, the event $(S^i)$
must occur. Thus, if the interval $[x_i, y_i]$ does not intersect the interval
$[\frac{\text{diam}(P)}{8}, \frac{\text{diam}(P)}{4}]$, then
$Pr[P_i \text{ non-terminally intersects } B(y,r)] = 0$.

We therefore turn to the case where $[x_i, y_i]$ does intersect the interval
$[\frac{\text{diam}(P)}{8}, \frac{\text{diam}(P)}{4}]$. Recall that in
defining the probability distribution dist$(diam(P),n)$, we have divided
the latter interval into $\log n$ subintervals $I_1, I_2, \ldots, I_{\log n}$ 
of equal length. Denote by $a_l$ the probability 
\[ Pr[\text{ a random sample drawn from dist}(\text{diam}(P),n) \text{ belongs to } I_{l}].\] For convenience,
define $I_{\log n + 1} = [\text{diam}(P)/4, \infty)$ and $a_{\log n + 1} = 0$. 

Let $I_{l_i}$ be the subinterval that contains $x_i$. (In case $x_i < \frac{\text{diam}(P)}{8}$, let $l_i = 1$.) The length of $[x_i,y_i]$ is at most $2r$, $2r \leq \text{diam}(P)/{8\log n}$, and the length
of each of the subintervals is $\diam(P)/8 \log n$. Thus
$[x_i,y_i]$ can intersect at most one more subinterval, and this is  $I_{l_i + 1}$.
Let $r_1$ and $r_2$ be the length of $I_{l_i} \cap [x_i,y_i]$ and $I_{{l_i}+1} \cap [x_i,y_i]$ respectively. Note that $r_1+r_2 \leq y_i-x_i\leq 2r$.




To bound $Pr[P_i \text{ non-terminally intersects } B(y,r)]$, we now have
two cases. We say that $p_i$ is {\em far} (from the ball $B(y,r)$) if
$l_i \in \{\log n -1, \log n\}$. We say that $p_i$ is {\em near} 
if $1 \leq l_i \leq \log n - 2$.
\\\\\textbf{Case 1: $p_i$ is far.} In this case $a_{l_i}, a_{l_i + 1} \leq
\frac{2}{n}.$ Thus
\begin{eqnarray*}
Pr[S^i]& \leq & Pr[\beta_i \text{ lands in } I_{l_i} \cap [x_i,y_i]]+Pr[\beta_i \text{ lands in } I_{{l_i}+1} \cap [x_i,y_i]]\\
& \leq &  \frac{r_1}{\text{diam}(P)/{8\log n}} a_{l_i} + \frac{r_2}{\text{diam}(P)/{8\log n}} a_{l_i + 1}  \\
& \leq & \frac{2r}{\text{diam}(P)/{8\log n}} \cdot \frac{2}{n} \\
& = & \frac{32 r \log n}{n \cdot \text{diam}(P)}.
\end{eqnarray*}

Thus, $Pr[P_i \text{ non-terminally intersects } B(y,r)] \leq Pr[S^i] \leq
\frac{32 r \log n}{n \cdot \text{diam}(P)}.$
\\\\\textbf{Case 2: $p_i$ is near.}  For such a $p_i$ we have the following
crucial observation. 

\begin{claim}
$Pr[P_i \text{ non-terminally intersects } B(y,r)] \leq \frac{32 r \log n}{\diam(P)} 
Pr[P_i \text{ terminally intersects } B(y,r)]$.
\end{claim}

\begin{proof} 
Suppose that $P_1, P_2, \ldots, P_{i-1}$ have been chosen and $B(y,r) 
\subseteq P_1 \cup P_2 \cup \cdots \cup P_{i-1}$. Conditioned on such
a history, we have $Pr[P_i \text{ non-terminally intersects } B(y,r)]
= Pr[P_i \text{ terminally intersects } B(y,r)] = 0$ and the
claimed inequality holds.

Now suppose that $B(y,r) 
\setminus (P_1 \cup P_2 \cup \cdots \cup P_{i-1}) \neq \emptyset$. Let
us condition on such a history. Then,  $P_i$ terminally intersects
$B(y,r)$ if $\beta_i$ lands in $I_{l_i + 2} \cup I_{l_i + 3} \cup \cdots
\cup I_{\log n}$. Thus, using Observation \ref{obs:sumeq},

\[Pr[P_i \text{ terminally intersects } B(y,r)] \geq
a_{l_i + 2} + a_{l_i + 3} + \cdots + a_{\log n} = a_{l_i + 1}.\]

On the other hand, 
\begin{eqnarray*}
Pr[P_i \text{ non-terminally intersects } B(y,r)] & \leq &
Pr[S^i \ | \ \beta_i \in I_{l_i} \cup I_{l_i + 1}] \cdot Pr[\beta_i \in I_{l_i} \cup I_{l_i + 1}] \\
& \leq & \left( \frac{2}{3}\frac{r_1}{\diam(P)/8 \log n} +
                \frac{1}{3}\frac{r_2}{\diam(P)/8 \log n} \right) \cdot (a_{l_i}
                 + a_{l_i + 1}) \\
& \leq & \frac{32 \cdot r \log n}{3\cdot \diam(P)} \cdot 3 a_{l_i + 1} \\
& \leq & \frac{32 \cdot r \log n}{\diam(P)} \cdot Pr[P_i \text{ terminally intersects } B(y,r)].
\end{eqnarray*}
\end{proof}

Hence the expected number of subsets that intersect $B(y,r)$ non-terminally is
\begin{eqnarray*}
& \leq &  \sum_{i=1}^p Pr[P_i \text{ non-terminally intersects } B(y,r)]\\
     & \leq &  \sum_{i: p_i \text{ is far}}Pr[P_i \text{ non-terminally intersects } B(y,r)] \\
     &      & \hspace{0.2in} + \sum_{i:p_i \text{ is near}} Pr[P_i \text{ non-terminally intersects } B(y,r)] \\
     & \leq &  \frac{32 r \log n}{\text{diam}(P)} \sum_{i=1}^p \frac{1}{n} +   
              \frac{32r\log n}{\text{diam}(P)} \sum_{i=1}^p Pr[P_i \text{ terminally intersects } B(y,r)] \\
     & \leq &  c \frac{r}{\diam(P)} \log n.
\end{eqnarray*}
For the last inequality, we used the fact that $\sum_{i=1}^p Pr[P_i \text{ terminally intersects } B(y,r)] = 1$, since there is exactly one $P_i$ that terminally
intersects $B(y,r)$. Putting the two cases together, we have 
\[
 E[T] \leq 1+c \frac{r}{\diam(P)} \log n.
\]
%
%
%
\end{proof}

We conclude by summarizing the result.

\begin{theorem}
\label{thm:partitioning}
Let $Z$ be a point set with an associated metric $d$, let $P \subseteq Z$
be a point set with at least $2$ points, and $n \geq |P|$ be a parameter. There is a polynomial-time probabilistic algorithm RAND-PARTITION$(P)$ that partitions $P$ into blocks $\{P_1, P_2, \ldots, P_t\}$ and has the
following guarantees:
\begin{enumerate}
 \item For each $1 \leq i \leq t$, $diam(P_i) \leq$ $diam(P)/2$.
 
 \item There is a constant $c > 0$ so that for any ball $B$ (centered at some point in $Z$) of radius $r \leq \frac{diam(P)}{16 \log n}$,
the expected size of the set $\{ i | P_i \cap B \neq \emptyset \}$ is
at most $1 + c \frac{r}{diam(P)} \log n$ and the expected number of blocks that non-terminally intersect $B$ is at most $c\frac{r}{\text{diam}(P)} \log n$. 
\end{enumerate}
\end{theorem}


\section{Algorithm for MCC}\label{sec:MCC}
We now describe our $(1+\epsilon)$-factor approximation algorithm for the MCC problem. Recall that we are given a set $X$ of clients, a set $Y$ of servers,
and a metric $d$ on $X \cup Y$. We wish to compute a cover for $X$ 
with minimum cost. Let $m=|Y|$ and $n=|X|$.

For $P \subseteq X$, let $\opt(P)$ denote some optimal cover for $P$. Denote
by $\cost(B)$ the cost of a ball $B$ (the $\alpha$-th power of B's radius)
and by $\cost(\mathcal{B})$ the cost $\sum_{B \in \mathcal{B}}\cost(B)$ of a set $\mathcal{B}$ of balls. 

To compute a cover for $P$, our algorithm first guesses the set $\Q \subseteq 
\opt(P)$ consisting of all the large balls in $\opt(P)$. As we note in the structure lemma below,
we may assume that the number of large balls in $\opt(P)$ is small. We then
use the algorithm of Theorem~\ref{thm:partitioning} to partition $P$ into
$\{P_1, P_2, \ldots, P_t\}$. For each $1 \leq i \leq t$, we recursively
compute a cover for the set $P'_i \subseteq P_i$ of points not covered 
by $\Q$. 

To obtain an approximation guarantee for this algorithm, we use the
guarantees of Theorem~\ref{thm:partitioning}. With this overview, we
proceed to the structure lemma and a complete description of the algorithm.

\subsection{A Structure Lemma}
It is not hard to show that for any $\gamma \geq 1$ and $P \subseteq X$ such
that $\diam(P)$ is at least a constant factor of $\diam(X\cup Y)$, $\opt(P)$ 
contains at most $(c/\gamma)^{\alpha}$ balls of radius at least $\diam(P)/\gamma$. Here $c$ is some absolute constant. The following structural lemma extends
this fact.
\begin{lemma}\label{lem:structure}
Let $P \subseteq X$, $0 < \lambda < 1$ and $\gamma \geq 1$,  and suppose that 
$\opt(P)$ does not contain any ball of radius greater than or equal to
$2 \alpha \cdot\diam(P)/\lambda$. Then the number of balls in $\opt(P)$ of radius 
greater than or equal to $\diam(P)/\gamma$ is at most $c(\lambda, \gamma) :=
(9 \alpha \gamma / \lambda)^{\alpha}$.
\end{lemma}
\begin{proof}
Suppose that $\opt(P)$ does not contain any ball of radius greater than or 
equal to $2 \alpha \cdot\diam(P)/\lambda$. Note that each ball in $\opt(P)$ intersects
$P$ and has radius at most $2 \alpha \cdot\diam(P)/\lambda$. Thus the point  set
$\{ z \in X\cup Y \ | \ z \in B \mbox{ for some } B \in \opt(P)\}$ has diameter
at most $\diam(P) + 8 \alpha \cdot\diam(P)/\lambda \leq 9 \alpha \cdot\diam(P)/\lambda$. It
follows that there is a ball centered at a point in $Y$, with radius
at most $9 \alpha \cdot\diam(P)/\lambda$ that contains $P$.

Let $t$ denote the number of balls in $\opt(P)$ of radius 
greater than or equal to $\diam(P)/\gamma$. By optimality of $\opt(P)$,
we have $t \cdot (\diam(P)/\gamma)^{\alpha} \leq (9 \alpha \cdot\diam(P)/\lambda)^{\alpha}$.
Thus $t \leq (9 \alpha \gamma / \lambda)^{\alpha}$.
\end{proof}
\subsection{The Algorithm}\label{subsec:algo}
We may assume that the minimum distance between two points in $X$ is $1$. Let
$L = 1+\log (\diam(X))$. As we want a $(1 + \eps)$-approximation, we fix a 
parameter $\lambda = \eps/2L$. Let $\gamma = \frac{c \log n}{\lambda}$, where
$c$ is the constant in Theorem~\ref{thm:partitioning}. Denote $\mathcal{D}$ to be the set of balls such that each ball is centered at a point of $y\in Y$ and has radius $r=d(x,y)$ for some $x\in X$. We note that for any $P \subseteq X$, 
any ball in $\opt(P)$ must belong to this set. Note that $|\mathcal{D}|\leq mn$.
Recall that $c(\lambda, \gamma) = (9 \alpha \gamma / \lambda)^{\alpha}$.

With this terminology, the procedure POINT-COVER($P$) described as Algorithm~\ref{alg:cover} returns a cover of $P \subseteq X$. If $|P|$ is smaller than
some constant, then the procedure returns an optimal solution by searching
all covers with a constant number of balls. In the general case, one candidate
solution is the best single ball solution. For the other candidate solutions,
the procedure first computes a partition $\{P_1,\ldots,P_{\tau}\}$ of $P$, 
using the RAND-PARTITION$(P)$ procedure. Here RAND-PARTITION$(P)$ is
called with $Z = X \cup Y$ and $n = |X| \geq |P|$. Then it iterates over all possible subsets of $\mathcal{D}$ of size at most $c(\lambda,\gamma)$ containing balls of radius greater than $\diam(P)/\gamma$. For each such subset $\Q$ and
$1 \leq i \leq \tau$, it computes the set $P_i' \subseteq P_i$ of points not
covered by $\Q$. It then makes recursive calls and generates the candidate 
solution $\Q \cup \bigcup^{\tau}_{\text{i=1}} \mbox{POINT-COVER}(P'_{\textit{i}})$. Note that
all the candidate solutions are actually valid covers for $P$. Among
these candidate solutions the algorithm returns the best solution.  
\begin{algorithm}[h]
    \caption{POINT-COVER$(P)$}
    \label{alg:cover}
    \begin{algorithmic}[1]
        \REQUIRE A subset $P \subseteq X$.
        \ENSURE A cover of the points in $P$.
	\IF{$|P|$ is smaller than some constant $\kappa$}
	    \RETURN a minimum solution by checking all covers with at most $\kappa$ balls. 
	\ENDIF
        \STATE sol $\leftarrow$ the best cover with one ball
        \STATE cost $\leftarrow \cost(sol)$	
        \STATE Let $\{P_1,\ldots,P_{\tau}\}$ be the set of nonempty subsets returned by RAND-PARTITION$(P)$ 
        \STATE Let $\mathcal{B}$ be the set of balls in $\mathcal{D}$ having radius greater than $\frac{\diam(P)}{\gamma}$
        \FOR{\textbf{each} $\Q \subseteq \mathcal{B}$ of size at most $c(\lambda,\gamma)$}
	    \FOR {$i=1$ to $\tau$}
		\STATE Let $P'_i = \{p \in P_i \ | \ p \not\in \bigcup_{B \in \Q} B\}$ 
	    \ENDFOR
	    \STATE $\Q' \leftarrow \Q \cup \bigcup_{\text{i=1}}^{\tau}$ POINT-COVER$(P'_i)$
	    \IF{$\cost(\Q') <$ cost}
		\STATE cost $\leftarrow \cost(\Q')$
		\STATE sol $\leftarrow \Q'$
	    \ENDIF
	\ENDFOR
        \RETURN sol
    \end{algorithmic}
\end{algorithm}

Our overall algorithm for MCC calls the procedure POINT-COVER$(X)$ to get a 
cover of $X$. 

\subsection{Approximation Guarantee}
For $P \subseteq X$, let $\level(P)$ denote the smallest non-negative integer
$i$ such that $\diam(P) < 2^i$. As the minimum interpoint distance in $X$ is
$1$, $\level(P) = 0$ if and only if $|P| \leq 1$. Note that $\level(X) \leq L$.

The following lemma bounds the quality of the approximation of our algorithm.

\begin{lemma}\label{lem:cover}
POINT-COVER$(P)$ returns a solution whose expected cost is at most $(1+\lambda)^{l} \cost(\opt(P))$, where $l = \level(P)$. 
\end{lemma}

\begin{proof}
We prove this lemma using induction on $l$. If $l=0$, then $|P| \leq 1$ and POINT-COVER$(P)$ returns an optimal solution, whose cost is $\cost(\opt(P))$. Thus assume that $l \geq 1$ and the statement is true for subsets having level at most $l-1$. Let $P \subseteq X$ be a point set with $\level(P) = l$. If $|P|$ is
smaller than the constant threshold $\kappa$, POINT-COVER$(P)$ returns an optimal
solution. So we may assume that $|P|$ is larger than this threshold. We have
two cases.
\\\\\textbf{Case 1:} There is some ball in $\opt(P)$ whose radius is 
at least $2 \alpha \cdot\diam(P)/\lambda$. Let $B$ denote such a ball and 
$r(B) \geq 2 \alpha \cdot\diam(P)/\lambda$ be its radius. Since $(1 + \lambda/2 \alpha) r(B)
\geq r(B) + \diam(P)$, the concentric ball of radius $(1 + \lambda/2 \alpha) r(B)$ contains $P$. It follows that there is a cover for $P$ that consists of
a single ball and has cost at most 
\[(1 + \lambda/2 \alpha)^{\alpha} r(B)^{\alpha} \leq (1 + \lambda) 
\cost(\opt(P)) \leq (1 + \lambda)^l \cost(\opt(P)).\]
\textbf{Case 2:} There is no ball in $\opt(P)$ whose radius is 
at least $2 \alpha \cdot\diam(P)/\lambda$. Let $\Q_{\text{0}} \subseteq \opt(P)$ contain
those balls of radius at least $\diam(P)/\gamma$. It follows from
Lemma~\ref{lem:structure} that $|\Q_{\text{0}}| \leq \textit{c}(\lambda,\gamma)$. Thus
the algorithm considers a $\Q$ with $\Q = \Q_{\text{0}}$. Fix this iteration. Also
fix the partition $\{P_1,\ldots,P_{\tau}\}$ of $P$ computed by 
RAND-PARTITION$(P)$. RAND-PARTITION ensures that $\diam(P_i) \leq \diam(P)/2$ for $1\leq i\leq \tau$. Thus $\diam(P'_i) \leq \diam(P)/2$ and the level of each $P'_i$ is at most $l-1$. Hence by induction the expected value of 
$\cost(\text{POINT-COVER}(P'_i))$ is at most  $(1+\lambda)^{l-1} \cost(\opt(P'_i))$. 

Let $\Sp = \opt(P) \setminus \Q_{\text{0}}$. We argue below that the expected
value of $\sum\limits_{i=1}^{\tau} \cost(\opt(P'_i))$ is at most $(1 + \lambda) \cost(\Sp)$. Assuming this, we have

\begin{eqnarray*}
E[ \cost(\Q_{\text{0}} \cup \bigcup_{\textit{i=}\text{1}}^{\tau} \mbox{POINT-COVER}(P'_{\textit{i}})) ] & \leq &
\cost(\Q_{\text{0}}) + (\text{1}+\lambda)^{{\textit{l}}-\text{1}} \textit{E}[ \sum_{{\textit{i}}=\text{1}}^{\tau} \cost(\opt(P'_{\textit{i}})) ] \\
& \leq & \cost(\Q_{\text{0}}) + (\text{1}+\lambda)^{\textit{l}} \cost(\Sp) \\
& \leq & (\text{1}+\lambda)^l \cost(\opt(P)).
\end{eqnarray*} 
Thus POINT-COVER$(P)$ returns a solution whose expected cost is at most
$(1 + \lambda)^l \cost(\opt(P))$, as desired.

We now argue that the expected
value of $\sum_{i=1}^{\tau} \cost(\opt(P'_i))$ is at most $(1 + \lambda) \cost(\Sp)$. Let $\B_{\textit{i}}$ consist of those balls in $\Sp$ that intersect $P_i$. For
$B \in \Sp$, let $\mu(B)$ denote the number of blocks in the partition
$\{P_1,\ldots,P_{\tau}\}$ that $B$ intersects. Because $\B_{\textit{i}}$ is a cover
for $P'_i$, we have $\cost(\opt(P'_i)) \leq \cost(\B_{\textit{i}})$. Thus
\[ \sum_{i=1}^{\tau} \cost(\opt(P'_i)) \leq \sum_{i=1}^{\tau} \cost(\B_{\textit{i}})
= \sum_{\textit{B} \in \Sp} \mu(\textit{B}) \cost(\textit{B}). \]

By definition of $\Q_{\text{0}}$, any ball $B \in \Sp = \opt(P) \setminus \Q_{\text{0}}$ has
radius at most $\frac{\diam(P)}{\gamma} = \frac{\lambda \cdot\diam(P)}{c \log n}$,
where $c$ is the constant in Theorem~\ref{thm:partitioning}. We may assume 
that $c \geq 16$ and hence $\frac{\lambda \cdot\diam(P)}{c \log n} \leq
\frac{\diam(P)}{16 \log n}$. Theorem~\ref{thm:partitioning} now implies
that 
\[E[\mu(B)] \leq 1 + \frac{c \cdot r(B) \log n}{\diam(P)} \leq
  1 + \frac{c \log n}{\diam(P)}\cdot \frac{\lambda \cdot\diam(P)}{c \log n} =
  1 + \lambda.\]
Thus the expected value of $\sum_{i=1}^{\tau} \cost(\opt(P'_i))$ is
at most
\[ \sum_{B \in \Sp} E[\mu(B)] \cost(B) \leq (1 + \lambda) 
   \sum_{B \in \Sp} \cost(B) = (1 + \lambda) \cost(\Sp),\]
as claimed.
\end{proof}

We conclude that the expected cost of the cover returned by
POINT-COVER$(X)$ is at most $(1 + \lambda)^L \cost(\opt(X)) \leq
(1 + \eps) \cost(\opt(X))$, since $\lambda = \eps/2L$.

Now consider the time complexity of the algorithm. POINT-COVER$(P)$ makes 
$(mn)^{O(c(\lambda,\gamma))}$ direct recursive calls on subsets of diameter at most 
$\diam(P)/2$. Thus the overall time complexity of POINT-COVER$(X)$ 
can be bounded by $(mn)^{O(c(\lambda,\gamma) L)}$. Plugging in $\lambda = \eps/2L$,
$\gamma = c \log n/\lambda$, and $c(\lambda,\gamma) = (9 \alpha \gamma / \lambda)^{\alpha}$, we conclude

\begin{theorem}\label{th:mcc}
There is an algorithm for MCC that runs in time 
$(mn)^{O(\frac{\alpha L^2 {\log} n}{{\epsilon}^2})^{\alpha} L} $ and returns a cover
whose expected cost is at most $(1 + \eps)$ times the optimal. Here $L$ is 
1 plus the logarithm of the aspect ratio of $X$, that is,
the ratio of the maximum and minimum interpoint 
distances in the client set $X$.
\end{theorem}

Using relatively standard techniques, which we omit here,
we can pre-process the input to ensure that the ratio of the maximum and minimum interpoint
distances in $X$ is upper bounded by a polynomial in $\frac{mn}{\eps}$. However, this affects the optimal solution by a factor of at most $(1+\epsilon)$.
After this pre-processing, we have $L = O(\log \frac{mn}{\eps})$. 
Using the algorithm in Theorem \ref{th:mcc} after the pre-processing, we obtain a $(1+\epsilon)$ approximation with the quasi-polynomial running time $O(2^{\log^{O(1)} mn})$.
Here the $O(1)$ hides a constant that depends on $\alpha$
and $\eps$.


\section{Algorithm for $k$-clustering}\label{sec:clustering}
Recall that in $k$-clustering we are given a set $X$ of points, a metric $d$ on $X$, and a positive integer $k$. Let $|X|=n$. 
For $P \subseteq X$ and integer $\kappa \geq 0$, let $\opt(P,\kappa)$ denote an optimal solution of $\kappa$-clustering for $P$ (using balls whose center can be any point in $X$). We reuse the notions of $\level(P)$, $\cost(B)$ and $\cost(\mathcal{B})$ from Section \ref{sec:MCC}, for a point set $P$, a ball $B$, and a set $\mathcal{B}$ of balls, respectively. Denote $\mathcal{D}$ to be the set of balls such that each ball is centered at a point of $y\in X$ and has radius $r=d(x,y)$ for some $x\in X$. We note that for any $P \subseteq X$, 
any ball in $\opt(P,\kappa)$ must belong to this set. Note that $|\mathcal{D}|\leq n^2$.  

To start with we prove a structure lemma for $k$-clustering.

\begin{lemma}\label{lem:structure-cluster}
Let $P \subseteq X$, $\kappa$ be a positive integer, and $\gamma \geq 1$. Then the number of balls in $\opt(P,\kappa)$ of radius 
greater than or equal to $\diam(P)/\gamma$ is at most $c(\gamma) :=
\gamma^{\alpha}$.
\end{lemma}

\begin{proof}
Note that any ball centered at a point in $P$ and having radius $\diam(P)$ contains all the points of $P$. Now by definition of $\diam(P)$ and $\mathcal{D}$, there is a point $x \in P$ such that the ball $B(x,\diam(P)) \in \mathcal{D}$. Hence $\opt(P,\kappa) \leq \diam(P)^{\alpha}$.

Let $t$ denote the number of balls in $\opt(P,\kappa)$ of radius 
greater than or equal to $\diam(P)/\gamma$. By optimality of $\opt(P,\kappa)$,
we have $t\cdot (\diam(P)/\gamma)^{\alpha} \leq \diam(P)^{\alpha}$.
Thus $t \leq \gamma^{\alpha}$.
\end{proof}

Like in the case of MCC, we assume that the minimum distance between two points in $X$ is $1$. Let
$L = 1+\log (\diam(X))$. We fix a 
parameter $\lambda = \eps/6L$. Let $\gamma = \frac{c \log n}{\lambda}$, where
$c$ is the constant in Theorem~\ref{thm:partitioning}. 

We design a procedure CLUSTERING$(P,\kappa)$ (see Algorithm \ref{alg:cluster}) that given a subset $P$ of $X$ and an integer $\kappa$, returns a set of at most $(\text{1}+3\lambda)^{\textit{l}}\kappa$ balls whose union contains $P$, where $l=\level(P)$. We overview this procedure, focussing on the differences from the
procedure POINT-COVER$()$ used to solve the MCC problem. In CLUSTERING$(P,\kappa)$, RAND-PARTITION$(P)$ is called with $Z = X$ and $n = |X| \geq |P|$. We require two properties of the partition $\{P_1,\ldots,P_{\tau}\}$ of $P$ computed by 
RAND-PARTITION$(P)$. Let $\Q_{\text{0}}$ be the set containing the large balls of $\opt(P,\kappa)$, that is, those with radius at least $\diam(P)/\gamma$. 
Let $\Sp = \opt(P,\kappa) \setminus \Q_{\text{0}}$ denote the set of small balls,
and let $\Sp_i \subseteq \Sp$ consist of those balls that contain at least one point in $P_i$ that is not covered by $\Q_{\text{0}}$. We would like (a) $\sum_{i=1}^{\tau} \cost(\Sp_i) \leq (1 + 3 \lambda) \cost(\Sp)$, and (b) 
$\sum_{i=1}^{\tau} |\Sp_i| \leq (1 + 3 \lambda) |\Sp|$.  Theorem~\ref{thm:partitioning} ensures that each of (a) and (b) holds in expectation. However, we would like both (a) and (b) to hold simultaneously, not just in expectation. For this
reason, we try $\Theta(\log n)$ independent random partitions in Line~\ref{lin:many-partitions}, ensuring that 
with high probability, properties (a) and (b) hold for at least one of them.
 
Now let us fix one of these $\Theta(\log n)$ trials where we got a partition $\{P_1,\ldots,P_{\tau}\}$ satisfying properties (a) and (b), and also fix an iteration in Line~\ref{lin:guess-Q} where we have $\Q = \Q_{\text{0}}$. Let $P'_i \subseteq P_i$ be the points not covered by $\Q_{\text{0}}$. For each $1 \leq i \leq \tau$ and $0 \leq \kappa_1 \leq (1 + 3 \lambda)\kappa$, we set cluster($P'_i,\kappa_1$) to be the cover obtained by recursively invoking CLUSTERING$(P'_i,\kappa_1)$ (as in Line~\ref{lin:recurse}).

Let us call a
tuple $(\kappa_1, \kappa_2, \ldots, \kappa_{\tau})$ of integers {\em valid} if
$0 \leq \kappa_i \leq (1 + 3 \lambda) (\kappa - |\Q_{\text{0}}|)$ and 
$\sum_{i = 1}^{\tau} \kappa_i \leq (1 + 3 \lambda) (\kappa - |\Q_{\text{0}}|)$. We would like to minimize $\sum_{i = 1}^{\tau} \cost(\text{cluster}(P'_i,\kappa_i))$ over all valid tuples $(\kappa_1, \kappa_2, \ldots, \kappa_{\tau})$. As there are too many valid tuples to allow explicit enumeration, we solve this optimization problem in Lines \ref{lin:dp-begin}--\ref{lin:dp-end} via a dynamic programming approach.

This completes our overview. Our overall algorithm for $k$-clustering calls the procedure CLUSTERING$(X,k)$. Next we give the approximation bound on the cost of the solution returned by CLUSTERING$(P,\kappa)$. 

\begin{algorithm}[h]
    \caption{CLUSTERING$(P,\kappa)$}
    \label{alg:cluster}
    \begin{algorithmic}[1]
        \REQUIRE A subset $P \subseteq X$, an integer $\kappa$.
        \ENSURE A set of balls whose union contains the points in $P$.
	\IF{$|P|$ is smaller than some constant $\beta$}
	    \RETURN a minimum solution by checking all solutions with at most $\min\{\kappa,\beta\}$ balls. 
	\ENDIF
        \STATE sol $\leftarrow$ the best solution with one ball
        \STATE cost $\leftarrow \cost(sol)$
        \STATE $l \leftarrow$ level($P$)
	\FOR{\textbf{all} $2\log_{3/2} {n}$ iterations} \label{lin:many-partitions}
        \STATE Let $\{P_1,\ldots,P_{\tau}\}$ be the set of nonempty subsets returned by RAND-PARTITION$(P)$ 
        \STATE Let $\mathcal{B}$ be the set of balls in $\mathcal{D}$ having radius greater than $\frac{\diam(P)}{\gamma}$
        \FOR{\textbf{each} $\Q \subseteq \mathcal{B}$ of size at most $c(\gamma)$} \label{lin:guess-Q}
	    \FOR {\textit{i} $=1$ to $\tau$}
		\STATE Let $P'_i = \{p \in P_i \ | \ p \not\in \bigcup_{B \in \Q} B\}$ 
	    \ENDFOR
	    \FOR {\textbf{each} $1\leq$ \textit{i} $\leq \tau$ \text{and} $0\leq \kappa_1\leq (\text{1+3}\lambda)\kappa$}
		\STATE cluster($P'_i,\kappa_1) \leftarrow$ CLUSTERING$(P'_i,\kappa_1)$ \label{lin:recurse}
	    \ENDFOR
	    \FOR {\textit{i} $=0$ to $\tau-1$} \label{lin:dp-begin}
		\STATE $R_i \leftarrow \bigcup_{j=i+1}^{\tau} P'_j$
	    \ENDFOR
	    \FOR {$\kappa_1=0$ to $(\text{1+3}\lambda)\kappa$} 
		\STATE cluster($R_{\tau-1},\kappa_1$) $\leftarrow$ cluster($P'_{\tau},\kappa_1)$
	    \ENDFOR	    
	    \FOR {\textbf{all} \textit{i} $=\tau-2$ to $0$ \text{and} $0\leq \kappa_1\leq (\text{1+3}\lambda)\kappa$}		
		\STATE  $\kappa'_{\min} \leftarrow$ $ \argmin_{\kappa': 0\leq \kappa'\leq \kappa_1}$ cost$($cluster($P'_{i+1},\kappa')$ $\cup$ cluster($R_{i+1},\kappa_1-\kappa'))$
		\STATE  cluster($R_{i},\kappa_1) \leftarrow$ cluster($P'_{i+1},\kappa'_{\min})$ $\cup$ cluster($R_{i+1},\kappa_1-\kappa'_{\min})$
	    \ENDFOR
	    \STATE $\Q' \leftarrow \Q$ $ \cup$ cluster$(R_0,(1 + 3\lambda)\cdot (\kappa - |\Q|))$ \label{lin:dp-end}
	    \IF{$|\Q'| \leq (\text{1+3}\lambda)^{\textit{l}}\kappa$ \text{and} $\cost(\Q') <$ cost}
		\STATE cost $\leftarrow \cost(\Q')$
		\STATE sol $\leftarrow \Q'$
	    \ENDIF
	    
	\ENDFOR
	\ENDFOR
        \RETURN sol
    \end{algorithmic}
\end{algorithm}

\begin{lemma}\label{lem:cluster}
For any $P \subseteq X$ and an integer $\kappa \geq 1$, CLUSTERING$(P,\kappa)$ returns a solution consisting of at most $(1+3\lambda)^{\textit{l}}\kappa$ balls and with probability at least $1-\frac{|P|-1}{n^2}$, the cost of the solution is at most $(1+3\lambda)^{l} \cost(\opt(P,\kappa))$, where $l = \level(P)$. 
\end{lemma}

\begin{proof}
We prove this lemma using induction on $l$. If $l=0$, then $|P| \leq 1$ and CLUSTERING$(P,$ $\kappa)$ returns an optimal solution, whose cost is $\cost(\opt(P,\kappa))$. Thus assume that $l \geq 1$ and the statement is true for subsets having level at most $l-1$. If $|P|$ is
smaller than the constant threshold $\beta$, CLUSTERING$(P,\kappa)$ returns an optimal
solution. So we may assume that $|P|$ is larger than this threshold. 

Consider one of the $2\log_{3/2} {n}$ iterations of CLUSTERING$(P,\kappa)$. Fix the partition $\{P_1,\ldots,P_{\tau}\}$ of $P$ computed by 
RAND-PARTITION$(P)$ in this iteration. Let $\Q_{\text{0}}$ be the set containing
the balls of $\opt(P,\kappa)$ with radius at least $\diam(P)/\gamma$. It follows from
Lemma~\ref{lem:structure-cluster} that $|\Q_{\text{0}}| \leq \textit{c}(\gamma)$. Fix the choice $\Q = \Q_{\text{0}}$. 

RAND-PARTITION($P$) ensures that $\diam(P_i) \leq \diam(P)/2$ for $1\leq i\leq \tau$. Thus $\diam(P'_i) \leq \diam(P)/2$ and the level of each $P'_i$ is at most $l-1$. Hence by induction, cluster($P'_i,\kappa_1)$ contains at most $(\text{1+3}\lambda)^{\textit{l}-1} \kappa_1$ balls and with probability at least $1-\frac{|P'_i|-1}{n^2}$, its cost 
is at most $(1+3\lambda)^{l-1} \cost(\opt(P'_i,\kappa_1))$ for $1\leq \kappa_1\leq (1+3\lambda)\kappa$. 


Let $\Sp = \opt(P,\kappa) \setminus \Q_{\text{0}}$. Thus $|\Sp|\leq \kappa - |\Q_{\text{0}}|$. We note that the union of the balls in $\Sp$ contains the points in $\cup_{i=1}^{\tau} P'_i$. Let $\Sp_{i} \subseteq \Sp$ be the set of balls that intersect with $P'_i$ and thus the union of balls in $\Sp_{\textit{i}}$ contains $P'_i$, where $1\leq i\leq \tau$. We argue below that the value of $\sum_{i=1}^{\tau} |\Sp_{\textit{i}}|$ is at most $(1 + 3\lambda) |\Sp|$ and the value of $\sum_{i=1}^{\tau} \cost(\opt(P'_i,|\Sp_{\textit{i}}|))$ is at most $(1 + 3\lambda) \cost(\Sp)$ with probability at least $1/3$. Then the probability is at least $1-\frac{1}{n^2}$ that corresponding to one of the $2{\log}_{3/2} n$ iterations the value of $\sum_{i=1}^{\tau} |\Sp_{\textit{i}}|$ is at most $(1 + 3\lambda) |\Sp|$ and the value of $\sum_{i=1}^{\tau} \cost(\opt(P'_i,|\Sp_{\textit{i}}|))$ is at most $(1 + 3\lambda) \cost(\Sp)$. We refer to this event as the first good event. Now let us assume that the first good event occurs. 

Note that by induction, cluster($P'_i,|\Sp_{\textit{i}}|)$ contains at most $(1+3\lambda)^{l-1} |\Sp_{\textit{i}}|$ balls and with probability at least $1-\frac{|P'_i|-1}{n^2}$, its cost is at most $(1+3\lambda)^{l-1} \cost(\opt(P'_i,|\Sp_{\textit{i}}|))$ for $1\leq i\leq \tau$. The probability that for every $1\leq i\leq \tau$, the cost of cluster($P'_i,|\Sp_{\textit{i}}|)$ is at most $(1+3\lambda)^{l-1} \cost(\opt(P'_i,|\Sp_{\textit{i}}|))$ and it contains at most $(1+3\lambda)^{l-1} |\Sp_{\textit{i}}|$ balls, is at least $$\prod_{i=1}^{\tau} 1-\frac{|P'_i|-1}{n^2} \geq \prod_{i=1}^{\tau} 1-\frac{|P_i|-1}{n^2} \geq 1-\frac{|P|-2}{n^2}.$$
We refer to this event as the second good event. Assuming this second good event occurs, it follows by an easy induction on $i$ that cluster($R_{i},\sum_{j=i+1}^{\tau}$ $ |\Sp_{\textit{j}}|)$ covers $R_i$=$\bigcup_{j=i+1}^{\tau} P'_j$ with at most $(\text{1}+3\lambda)^{\textit{l}-1}\sum_{j=i+1}^{\tau} |\Sp_{\textit{j}}|$ balls and its cost is at most $(\text{1+3}\lambda)^{{\textit{l}}-\text{1}} \sum_{{\textit{j}}=\text{i+1}}^{\tau} \cost(\opt(P'_{\textit{j}},|\Sp_{\textit{j}}|))$. Thus cluster($R_{0},\sum_{i=1}^{\tau} |\Sp_{\textit{i}}|)$ covers $R_0=\bigcup_{i=1}^{\tau} P'_i$ with at most $(\text{1}+3\lambda)^{\textit{l}-1}\sum_{i=1}^{\tau} |\Sp_{\textit{i}}|$ balls and its cost is at most $(\text{1+3}\lambda)^{{\textit{l}}-\text{1}} \sum_{{\textit{i}}=\text{1}}^{\tau} \cost(\opt(P'_{\textit{i}},|\Sp_{\textit{i}}|))$. Hence $\Q_{\text{0}} \cup$ cluster($R_{0},\sum_{i=1}^{\tau} |\Sp_{\textit{i}}|)$ covers $P$.
Now assuming both the good events occur, we have
\begin{eqnarray*}
|\Q_{\text{0}} \cup \text{cluster}(R_{0},\sum_{i=1}^{\tau} |\Sp_{\textit{i}}|)| & \leq & |\Q_{\text{0}}|+\sum_{\textit{i}=\text{1}}^{\tau} (\text{1+3}\lambda)^{\textit{l}-\text{1}} |\Sp_{\textit{i}}|\\ & \leq & |\Q_{\text{0}}|+(\text{1+3}\lambda)^{\textit{l}-\text{1}}\sum_{\textit{i}=\text{1}}^{\tau} |\Sp_{\textit{i}}|\\ & \leq & |\Q_{\text{0}}|+ (\text{1+3}\lambda)^{\textit{l}} |\Sp|\\ & \leq & |\Q_{\text{0}}|+ (\text{1+3}\lambda)^{\textit{l}} (\kappa - |\Q_{\text{0}}|)\\ & \leq & (\text{1+3}\lambda)^{\textit{l}} \kappa.
\end{eqnarray*}
\begin{eqnarray*}
\cost(\Q_{\text{0}} \cup \text{cluster}(R_{0},\sum_{i=1}^{\tau} |\Sp_{\textit{i}}|))  & \leq &
\cost(\Q_{\text{0}}) + (\text{1+3}\lambda)^{{\textit{l}}-\text{1}} \sum_{{\textit{i}}=\text{1}}^{\tau} \cost(\opt(P'_{\textit{i}},|\Sp_{\textit{i}}|))  \\
& \leq & \cost(\Q_{\text{0}}) + (\text{1+3}\lambda)^{\textit{l}} \cost(\Sp) \\
& \leq & (\text{1+3}\lambda)^l \cost(\opt(P,\kappa)).
\end{eqnarray*} 
The probability that both the good events occur is at least $$(1-\frac{|P|-2}{n^2})\cdot (1-\frac{1}{n^2})\geq 1-\frac{|P|-1}{n^2}.$$ 
Hence the statement of the lemma follows by noting that the value of $\sum_{i=1}^{\tau} |\Sp_{\textit{i}}|$ in $\text{cluster}(R_{0},\sum_{i=1}^{\tau} |\Sp_{\textit{i}}|)$ is at most $(1 + 3\lambda)\cdot (\kappa - |\Q_{\text{0}}|)$.

We now argue that the probability that $\sum_{i=1}^{\tau} |\Sp_{\textit{i}}|$ is greater than $(1 + 3\lambda) |\Sp|$ is at most $1/3$ and the probability that $\sum_{i=1}^{\tau} \cost(\opt(P'_i,|\Sp_{\textit{i}}|))$ is greater than $(1 + 3\lambda) \cost(\Sp)$ is at most $1/3$. Then using union bound, the probability that $\sum_{i=1}^{\tau} |\Sp_{\textit{i}}|$ is at most $(1 + 3\lambda) |\Sp|$ and $\sum_{i=1}^{\tau} \cost(\opt(P'_i,|\Sp_{\textit{i}}|))$ is at most $(1 + 3\lambda) \cost(\Sp)$ is at least $1/3$.

For
$B \in \Sp$, let $\mu(B)$ denote the number of blocks in the partition
$\{P_1,\ldots,P_{\tau}\}$ that non-terminally intersect $B$. We note that at most one block can terminally intersect $B$. Because $\Sp_{\textit{i}}$ is a cover for $P'_i$ with $|\Sp_{\textit{i}}|$ balls, we have $\cost(\opt(P'_i,|\Sp_{\textit{i}}|)) \leq \cost(\Sp_{\textit{i}})$. Thus
\[ \sum_{i=1}^{\tau} \cost(\opt(P'_i,|\Sp_{\textit{i}}|)) \leq \sum_{\textit{i}=\text{1}}^{\tau} \cost(\Sp_{\textit{i}})
= \sum_{\textit{B} \in \Sp} (\text{1}+\mu(\textit{B})) \cost(\textit{B})=\cost(\Sp)+\sum_{\textit{B} \in \Sp} \mu(\textit{B}) \cost(\textit{B}). \]
and,
\[
 \sum_{i=1}^{\tau} |\Sp_{\textit{i}}| \leq \sum_{\textit{B} \in \Sp} (\text{1}+\mu(\textit{B}))=|\Sp|+\sum_{\textit{B} \in \Sp} \mu(\textit{B}).
\]

By definition of $\Q_{\text{0}}$, any ball $B \in \Sp = \opt(P,\kappa) \setminus \Q_{\text{0}}$ has
radius at most $\frac{\diam(P)}{\gamma} = \frac{\lambda \cdot\diam(P)}{c \log n}$,
where $c$ is the constant in Theorem~\ref{thm:partitioning}. We may assume 
that $c \geq 16$ and hence $\frac{\lambda \cdot\diam(P)}{c \log n} \leq
\frac{\diam(P)}{16 \log n}$. Theorem~\ref{thm:partitioning} now implies
that 
\[E[\mu(B)] \leq \frac{c \cdot r(B) \log n}{\diam(P)} \leq
  \frac{c \log n}{\diam(P)}\cdot \frac{\lambda \cdot\diam(P)}{c \log n} =
  \lambda.\]
Using linearity of expectation, $$E[\sum_{B \in \Sp} \mu(B)]=\sum_{B \in \Sp} E[\mu(B)]\leq \lambda\cdot |\Sp|.$$ Now by Markov's inequality, $$Pr[\sum_{B \in \Sp} \mu(B) > 3\lambda\cdot |\Sp|] \leq \text{1/3}.$$ It follows that,
\begin{eqnarray*}
 & & Pr[\sum_{i=1}^{\tau} |\Sp_{\textit{i}}| > (\text{1 + 3}\lambda)|\Sp|]\\
 &\leq& Pr[|\Sp|+\sum_{\textit{B} \in \Sp} \mu(\textit{B})> (\text{1 + 3}\lambda)|\Sp|]\\
 &\leq& Pr[\sum_{B \in \Sp} \mu(B)>3\lambda\cdot |\Sp|]\leq \text{1/3}.
\end{eqnarray*}
Again using linearity of expectation,
$$E[\sum_{B \in \Sp} \mu(B)\cost(B)]=\sum_{B \in \Sp} E[\mu(B)]\cost(B)\leq \lambda\cdot \cost(\Sp).$$Now by Markov's inequality, $$Pr[\sum_{B \in \Sp} \mu(B)\cost(B) > 3\lambda\cdot \cost(\Sp)] \leq \text{1/3}.$$ It follows that,
\begin{eqnarray*}
 & & Pr[\sum_{i=1}^{\tau} \cost(\opt(P'_i,|\Sp_{\textit{i}}|)) > (\text{1 + 3}\lambda)\cost(\Sp)]\\
 &\leq& Pr[\cost(\Sp)+\sum_{\textit{B} \in \Sp} \mu(\textit{B}) \cost(\textit{B})> (\text{1 + 3}\lambda)\cost(\Sp)]\\
 &\leq& Pr[\sum_{B \in \Sp} \mu(B)\cost(B) > 3\lambda\cdot \cost(\Sp)] \leq \text{1/3}.
\end{eqnarray*}
\end{proof}

Since $\lambda = \eps/6L$, $(1+3\lambda)^L \leq 1+\eps$. Thus we conclude that with probability at least $1-\frac{1}{n}$, CLUSTERING$(X,k)$ returns a solution with at most $(\text{1}+\eps)k$ balls whose cost is at most $(1 + \eps) \cost(\opt(X,k))$.

Now consider the time complexity of the algorithm. CLUSTERING$(P,\kappa)$ makes 
$n^{O(c(\gamma))}$ direct recursive calls on subsets of diameter at most 
$\diam(P)/2$. Thus the overall time complexity of CLUSTERING$(X,k)$ 
can be bounded by $n^{O(c(\gamma) L)}$. Plugging in $\lambda = \eps/6L$,
$\gamma = c \log n/\lambda$, and $c(\gamma) = \gamma^{\alpha}$, we conclude

\begin{theorem}\label{th:cluster}
There is a randomized algorithm for $k$-clustering that runs in time 
$n^{O((\frac{L {\log} n}{{\epsilon}})^{\alpha} L)}$ and with probability at least $1-\frac{1}{n}$ returns a solution with at most $(1 + \eps)k$ balls
whose cost is at most $(1 + \eps)$ times the optimal. Here $L$ is 
1 plus the logarithm of the aspect ratio of $X$, that is,
the ratio of the maximum and minimum interpoint 
distances in the set $X$.
\end{theorem}


\section{Inapproximability Result}\label{sec:inapprox}
In this section we present an inapproximability result which complements the result in Section \ref{sec:MCC}. In particular here we consider the case when $\alpha$ is not a constant. The heart of this result is a reduction from the dominating set problem. Given a graph $G=(V,E)$, a dominating set for $G$ is a subset $V'$ of $V$ such that for any vertex $v \in V\setminus V'$, $v$ is connected to at least one vertex of $V'$ by an edge in $E$. The dominating set problem is defined as follows.\\\\
\textbf{Dominating Set Problem (DSP)}\\
INSTANCE: Graph $G=(V,E)$, positive integer $k \leq |V|$.\\
QUESTION: Is there a dominating set for $G$ of size at most $k$?\\

The following inapproximability result is proved by Kann \cite{viggo}.

\begin{theorem}\label{th:mds}
There is a constant $c > 0$ such that there is no polynomial-time $c\log |V|$-factor approximation algorithm for DSP assuming $\mathcal{P} \neq \mathcal{NP}$. 
\end{theorem}

The following theorem shows an inapproximability bound for MCC when $\alpha \geq \log |X|$. 

\begin{theorem}
For $\alpha \geq \log |X|$, no polynomial time algorithm for MCC can achieve an approximation factor better than $c \log |X|$ assuming $\mathcal{P} \neq \mathcal{NP}$. 
\end{theorem}

\begin{proof}
To prove this theorem we show a reduction from DSP. Given an instance $(G=(V,E),k)$ of DSP we construct an instance of MCC. The instance of MCC consists of two sets of points $X$ (clients) and $Y$ (servers), and a metric $d$ defined on $X\cup Y$. Let $V= \{v_1,v_2,\ldots,v_n\}$, where $n=|V|$. For each $v_i\in V$, $Y$ contains a point $y_i$ and $X$ contains a point $x_i$. For any point $p\in X\cup Y$, $d(p,p)=0$. For $i,j \in [n]$, $d(x_i,y_j)$ is 1 if $i=j$ or the edge $(v_i,v_j) \in E$, and $d(x_i,y_j)$ is 3 otherwise. For $i,j \in [n]$ such that $i \neq j$, we set $d(x_i,x_j) = d(y_i,y_j) = 2$. 

Consider two nonadjacent vertices $v_i$ and $v_j$. For any $x_t \in X$ such that $t\neq i,j$, $d(x_i,x_t)+d(x_t,y_j) \geq 3$. Similarly, for any $y_t \in Y$ such that $t\neq i,j$, $d(x_i,y_t)+d(y_t,y_j) \geq 3$. Thus $d$ defines a metric. Next we will prove that $G$ has a dominating set of size at most $k$ iff the cost of covering the points in $X$ using the balls around the points in $Y$ is at most $k$. 

Suppose $G$ has a dominating set $J$ of size at most $k$. For each vertex $v_j \in J$, build a radius 1 ball around $y_j$. We return this set of balls $\mathcal{B}$ as the solution of MCC. Now consider any point $x_i\in X$. If $v_i \in J$, then $x_i$ is covered by the ball around $y_i$. Otherwise, there must be a vertex $v_j \in J$ such that $(v_i,v_j)\in E$. Then $d(x_i,y_j)$ is 1 and $x_i$ is covered by the ball around $y_j$. Hence $\mathcal{B}$ is a valid solution of MCC with cost at most $k$.

Now suppose there is a solution $\mathcal{B}$ of MCC with cost at most $k$. If $k > |X|$, then $V$ is a dominating set for $G$ of size $|X| < k$. If $k \leq |X|$, our claim is that the radius of each ball in $\mathcal{B}$ is 1. Suppose one of the balls $B$ has a radius more than 1. Then the way the instance of MCC is created the radius should be at least 3. Hence $k\geq 3^{\alpha} \geq 3^{\log |X|} > |X|$, which is a contradiction. Now consider the set of vertices $J$ corresponding to the centers of balls in $\mathcal{B}$. It is not hard to see that $J$ is a dominating set for $G$ of size at most $k$.

Let OPT be the cost of any optimal solution of MCC for the instance $(X,Y,d)$. Then by the properties of this reduction the size of any minimum dominating set for $G$ is OPT. Thus if there is an approximation algorithm for MCC that gives a solution with cost $(c \log |X|) \cdot $OPT, then using the reduction we can produce a dominating set of size $(c \log |V|) \cdot $OPT. Then from Theorem \ref{th:mds} it follows that $\mathcal{P}=\mathcal{NP}$. This completes the proof of our theorem.
\end{proof}

\section{Conclusions}
\label{sec:concl}

One generalization of the MCC problem that has been studied \cite{CharikarP04,BiloCKK05} includes fixed costs for opening the servers. As input, we are given two point sets $X$ (clients) and $Y$ (servers), a metric on $Z = X \cup Y$, and a {\em facility cost} $f_y \geq 0$ for each server $y \in Y$. The goal is to find a subset $Y' \subseteq Y$, and a set of balls $\{B_y \ | y \in Y' \mbox{ and } B_y \mbox{ is centered at } y\}$ that covers $X$, so as to minimize $\sum_{y \in Y'} (f_y + r(B_y)^{\alpha})$. It is not hard to see that our approach generalizes in a straightforward way to give a $(1 + \eps)$ approximation to this problem using quasi-polynomial running time. To keep the exposition clear, we have focussed on the MCC rather than this generalization.     

The main open problem that emerges from our work is whether there one can obtain a $(1 + \eps)$-approximation for the $k$-clustering problem in quasi-polynomial time.
\\\\\textbf{Acknowledgements.} We would like to thank an anonymous reviewer of
an earlier version of this paper for suggestions that improved the guarantees
and simplified the proof of Theorem~\ref{thm:partitioning}. We also thank other reviewers for their feedback and pointers to the literature.
    
\bibliographystyle{plain}
\bibliography{paper106}

\end{document}